\documentclass[runningheads,a4paper]{llncs}
\usepackage{diagbox}
\usepackage{amsmath,amssymb}
\usepackage{graphicx}
\usepackage{caption}
\usepackage{subcaption}
\usepackage{color}
\usepackage{multirow}
\usepackage{url}
\usepackage{mathptmx}
\usepackage{dsfont}
\usepackage{times}
\usepackage{graphics}
\usepackage{rotating}
\usepackage{xspace}
\usepackage{colortbl}
\usepackage{tabularx}
\usepackage{dcolumn}
\usepackage{fancyhdr}
\usepackage{footmisc}
\usepackage{longtable}
\usepackage[usenames,dvipsnames]{xcolor}
\usepackage{enumitem}
\setlist{nolistsep}
\usepackage{algorithm}
\usepackage{algpseudocode}

\usepackage{amstext}
\usepackage{tikz}
\usetikzlibrary{arrows}

\usepackage[title]{appendix}

\usepackage{array}
\newcolumntype{L}[1]{>{\raggedright\let\newline\\\arraybackslash\hspace{0pt}}m{#1}}
\newcolumntype{C}[1]{>{\centering\let\newline\\\arraybackslash\hspace{0pt}}m{#1}}
\newcolumntype{R}[1]{>{\raggedleft\let\newline\\\arraybackslash\hspace{0pt}}m{#1}}
\begin{document}

\mainmatter

\title{Fast Simulation of Probabilistic Boolean Networks (Technical Report)}

\titlerunning{Fast Simulation of Probabilistic Boolean Networks}

\author{Andrzej Mizera \and Jun Pang \and Qixia Yuan
\thanks{Supported
by the National Research Fund, Luxembourg (grant 7814267). }
}
\authorrunning{Mizera, Pang and Yuan}

\institute{
Computer Science and Communications, University of Luxembourg, Luxembourg\\
\email{firstname.lastname@uni.lu}
}

%\author{Andrzej Mizera\inst{1} \and Jun Pang\inst{1,2} \and Qixia Yuan\inst{1}\thanks{Supported
%by the National Research Fund, Luxembourg (grant 7814267). }}
%\authorrunning{Mizera, Pang, and Yuan}
%\institute{
%Faculty of Science, Technology and Communication , University of Luxembourg
%\and
%Interdisciplinary Centre for Security, Reliability and Trust, University of Luxembourg\\
%\email{firstname.lastname@uni.lu}
%}

\maketitle

%------------------------------------------------------------------------------
\begin{abstract}
Probabilistic Boolean networks (PBNs) is an~important mathematical framework widely used for
modelling and analysing biological systems.
PBNs are suited
for modelling large biological systems, which more and more often arise in systems biology.
However, the large system size poses a~significant challenge to the analysis of PBNs, in
particular, to the crucial analysis of their steady-state behaviour.
Numerical methods for performing steady-state analyses suffer from the state-space explosion
problem, which makes the utilisation of statistical methods the only viable approach. However,
such methods require long simulations of PBNs, rendering the simulation speed a~crucial efficiency factor.
For large PBNs and high estimation precision requirements, a~slow simulation
speed becomes an~obstacle. In this paper, we propose a~structure-based method for fast
simulation of PBNs. This method first performs a~network reduction operation and then divides
nodes into groups for parallel simulation. Experimental results show that our method can lead to
an~approximately 10 times speedup for computing steady-state probabilities of a~real-life biological network.

%\keywords{Probabilistic Boolean networks, Markov chains, steady-state analysis, approximation}
\end{abstract}

%===============================================
\section{Introduction}
\label{sec:intro}
%===============================================
Systems biology aims to model and analyse biological systems from a~holistic perspective in order
to provide a~comprehensive, system-level understanding of cellular behaviour. Computational
modelling of a~biological system plays a~key role in systems biology. It connects the field of
traditional biology with mathematics and computational science, providing a~way to organize and
formalize available biological knowledge in a~mathematical model and to identify missing
biological information using formal means. Together with biochemical techniques, computational
modelling promotes the holistic understanding of real-life biological systems, leading to the
study of large biological systems. This brings a significant challenge to computational modelling
in terms of the state-space size of the system under study. Among the existing modelling
frameworks, probabilistic Boolean networks (PBNs) is well-suited for modelling large-size
biological systems. It is first introduced by Shmulevich et al.~\cite{SD10,TMPTSS13} as
a~ probabilistic generalisation of the standard Boolean networks (BNs) to model gene regulatory
networks (GRNs). The framework of PBNs incorporates rule-based dependencies between genes and
allows the systematic study of global network dynamics; meanwhile, it is capable of dealing with
uncertainty, which naturally occurs at different levels in the study of biological systems.

%Constructing a computational model allows systematisation of available biological knowledge concerning
%biochemical processes of a~biological system.
%Analysing the constructed model provides the possibility for identifying missing biological information
%and understanding the real-life system using formal means.
%This often leads to the reveal of directions for future experimental work that can provide data
%for better analysing and understanding the system under study.
%However, computational modelling of biological processes poses a significant challenge
%with respect to the state-space size of the system under study.
%It often arises that modelling of certain parts of cellular machinery such as gene regulatory networks (GRNs)
%or signal transduction pathways often leads to dynamical models characterised by huge state-spaces of
%sizes that surpass the sizes of any human-designed systems by orders of magnitude.
%Therefore, further development of efficient techniques and approaches for formal
%modelling and analysis of biological systems is required.

Focusing on the wiring of a~network, PBNs is essentially designed for revealing the long-run
(steady-state) behaviour of a~biological system. Comprehensive understanding of the long-run
behaviour is vital in many contexts. For example, attractors of a~gene regulatory network (GRN)
are considered to characterise cellular phenotypes~\cite{Kauffman69a}.
%One of the key aspects of analysing biological systems is the comprehensive understanding of their long-run (steady-state) behaviour,
%which is vital in many contexts, e.g., attractors of a gene regulatory network (GRN) were considered to characterise cellular phenotype~\cite{Kauffman69a}.
%in the analysis of the long-term influence of one gene on another
%gene in a~gene regulatory network (GRN)~\cite{SDZ02}.
There have been a~lot of studies in analysing the long-run behaviour of biological systems for
better understanding the influences of genes or molecules in the systems~\cite{SDZ02}.
Moreover, steady-state analyses have been used in gene intervention and external
control~\cite{SDZ02Gene,ADZ02}, which is of special interest to cancer therapists to predict the
potential reaction of a~patient to treatment. In the context of PBNs, many efforts have been
devoted to computing their steady-state probabilities.
%PBNs can be viewed as a discrete-time Markov chains (DTMCs)
%and the dynamics of PBNs is often studied in the realm of DTMCs
%so that the rich theories of DTMCs can be applied to the analysis of PBNs.
%Given a~PBN, one natural and crucial issue is to study the steady-state probabilities of its
%underlying DTMC, which characterise the long-run behaviour of the corresponding biological
%systems~\cite{Kau93}.
In~\cite{SGH03,TMPTS14}, efficient numerical methods are provided for computing the steady-state
probabilities of small-size PBNs. Those methods utilise an~important characteristics of PBNs,
i.e., a~PBN can be viewed as a~discrete-time Markov chain (DTMC) and its dynamics can be studied
with the use of the rich theory of DTMCs. The key idea of those methods relies on the computation
of the transition matrix of the underlying DTMC of the studied PBN. They perform well for
small-size PBNs. However, in the case of large-size PBNs, the state-space size becomes so huge
that the numerical methods are not scalable any more.

Many efforts are then spent on addressing the challenge of the huge state-space in large-size
PBNs.  In fact, the use of statistical methods and Monte Carlo methods remain the only feasible
approach to address the problem. In those methods, the simulation speed is an~important factor in
the performance of these approaches. For large PBNs and long trajectories, a~slow simulation speed
could render these methods infeasible as well. In our previous work~\cite{MPY15}, we have
considered the two-state Markov chain approach and the Skart method for approximate analysis of
large PBNs. Taking special care of efficient simulation, we have implemented these two methods in
the tool \textsf{ASSA-PBN}~\cite{assa} and successfully used it for the analysis of large PBNs
with a~few thousands of nodes. However, the required time cost is still expected to be reduced.
This requirement is of great importance for the construction of a~model, e.g., parameter
estimation, and for a~more precise and deep analysis of the system. In this work, we propose
a~structure-based method to speed up the simulation process. The method is based on analysing the
structure of a~PBN and consists of two key ideas: first, it removes the unnecessary nodes in the
network to reduce its size; secondly, it divides the nodes into groups and performs simulation for
nodes in a~group simultaneously. We show with experiments that our structure-based method can
significantly reduce the computation time for approximate steady-state analyses of large PBNs.
To the best of our knowledge, our proposed method is the first one to apply structure-based
analyses for speeding up the simulation of a~PBN.

\medskip\noindent
{\bf Structure of the paper.}
After presenting preliminaries in Section~\ref{sec:pre}, we describe our structure-based simulation method in Section~\ref{sec:method}.
We perform an~extensive evaluation and comparison of our method with the previous state-of-art methods in
Section~\ref{sec:evaluation} on a~large number of randomly generated PBNs and a~large real-life PBN model of apoptosis in hepatocytes.
We conclude our paper with some discussions in Section~\ref{sec:conclusion}.

%===============================================
\section{Preliminaries}
\label{sec:pre}
%===============================================

%----------------------------------------------------------------------------------
\subsection{Probabilistic Boolean networks (PBNs)}
\label{ssec:pbn}
%----------------------------------------------------------------------------------
A PBN $G(X,F)$ models elements of a biological system with a set of binary-valued nodes $X=\{x_1, x_2,\dots, x_n\}$.
For each node $x_i \in X$, the update of its value is guided by a set of \emph{predictor functions} $F_i=\{f_1^{(i)}, f_2^{(i)}, \dots, f_{\ell(i)}^{(i)}\}$,
where $\ell(i)$ is the number of predictor functions for node $x_i$.
Each $f_j^{(i)}$ is a~Boolean function whose inputs are a~subset of nodes, referred to as \emph{parent nodes} of $x_i$.
For each node $x_i$, one of its predictor functions will be selected to update the value of $x_i$ at each time point $t$.
This selection is in accordance with a~probability distribution
$C_i=(c_1^{(i)},c_2^{(i)},\dots,c_{\ell(i)}^{(i)})$, where the individual probabilities are
the \emph{selection probabilities} for the respective elements of $F_i$ and they sum to 1.
The value of node $x_i$ at time point $t$ is denoted as $x_i(t)$
and the state of the PBN at time point $t$ is denoted as $s(t)=(x_1(t),x_2(t),\dots,x_n(t)).$
The state space of the PBN is $S=\{0,1\}^n$ and it is of size $2^n$.
There are several variants of PBNs with respect to the selection of predictor functions and the
synchronisation of nodes update.
In this paper, we consider the \textit{independent synchronous} PBNs,
i.e., the choice of predictor functions for each node is made independently
and the values of all the nodes are updated synchronously.
The transition from state $s(t)$ to state $s(t+1)$ is performed by
randomly selecting a predictor function for each node $x_i$ from $F_i$
and by applying those selected predictor functions to update the values of all the nodes
synchronously.
We denote $f(t)$ the combination of all the selected predictor functions at time point $t$.
The transition of state $s(t)$ to $s(t+1)$ can then be denoted as $s(t+1)=f(t)(s(t)).$

Perturbations of a biological system are introduced by a perturbation rate $p \in (0,1)$ in a PBN.
The dynamics of a PBN is guided with both perturbations and predictor functions:
at each time point $t$, the value of each node $x_i$ is flipped with probability $p$;
and if no flip happens, the value of each node $x_i$ is updated with selected predictor functions synchronously.
Let $\boldsymbol{\gamma}(t)=(\gamma_{1}(t),\gamma_{2}(t),\dots,\gamma_{n}(t))$, where $\gamma_{i}(t)
\in \{0,1\}$ and $\mathbb{P}(\gamma_{i}(t)=1)=p$ for all $t$ and $i \in \{1,2,\dots,n\}$.
The transition of $s(t)$ to $s(t+1)$ in PBNs with perturbations is given by
\begin{equation}
s(t+1)=\begin{cases}
    s(t) \oplus \gamma(t)       & \quad \text{if } \gamma(t) \neq0\\
    f(t)(s(t)) & \quad \text{otherwise,}\\
  \end{cases}
  \label{equ:s}
\end{equation}
where $\oplus$ is the element-wise exclusive or operator for vectors.
According to Equation~\eqref{equ:s},
perturbations allow the system to move from a~state to any other
state in one transition, hence render the underlying Markov chain irreducible and aperiodic.
Thus, the dynamics of a~PBN with perturbations can be viewed as an~ergodic
DTMC~\cite{SD10}.
Based on the ergodic theory,
the long-run dynamics of a~PBN with perturbations
is governed by a~unique limiting distribution,
convergence to which is independent of the choice of the initial state.

The density of a~PBN is measured with its predictor function number and parent nodes number. For
a~PBN $G$, its density is defined as $\mathcal{D}(G)=\frac{1}{n}\sum_{i=1}^{\it N_F}\phi(i)$,
where $n$ is the number of nodes in $G$, ${\it N_F}$ is the total number of predictor functions in
$G$, and $\phi(i)$ is the number of parent nodes for the $i$th predictor function.

%----------------------------------------------------------------------------------
\subsection{Simulating a PBN}
\label{ssec:sim}
%----------------------------------------------------------------------------------
A PBN can be simulated via two steps based on its definition. First, perturbation is verified for
each individual node and a~node value is flipped if there is a~perturbation. Second, if no
perturbation happens for any of the nodes, the network state is updated by selecting predictor
functions for all the nodes and applying them. For efficiency reason, the selection of predictor
functions for each node $x_i$ is performed with the \emph{alias method}~\cite{WAJ77}, which allows
to make a~selection among choices in constant time irrespective of the number of choices. The
alias method requires the construction of an~\textit{alias table} of size proportional to the
number of choices, based on the selection probabilities of $C_i$.

%\begin{algorithm}[!t]
%\caption{One step simulation of a PBN}
%\label{alg:sim}
%\begin{algorithmic}[1]
%\Procedure{simulation}{$n, A,F,parents,p,s$}
%	\State $perturbation=false$;	
%	\State $i=0$;
%    \Repeat
%    \If{$rand()>p$}
%    \State $s.flip(i)$;
%    \State $perturbation=true;$
%    \State $i++$;
%    \EndIf
%	\Until $i=n-1$
%	\If{$perturbation$}
%	\State Return $s$;
%	\Else ~~//use predictor functions to update the state
%	\State $i=0;$
%	\Repeat
%	\State $index=A[i].next()$;
%	\State $l\_F=F.get(index);$
%	\State $l\_parents=parents.get(index);$~~//get the parent nodes of this function
%	\State $indexF=compute(l\_parents,s);$~~//compute the index for the given parents
%	\State $s'.set(i,l\_F[indexF])$
%	\Until $i=n-1$
%	\EndIf
%    \State Return $s'$.
%\EndProcedure
%\end{algorithmic}
%\end{algorithm}

%===============================================
\section{Structure-based Parallelisation}
\label{sec:method}
%===============================================
The simulation method described in the above section requires to check perturbations,
make a selection and perform updating a node for $n$ times in each step.
%The algorithm described in Algorithm~\ref{alg:sim} contains two loops of length $n$.
In the case of large PBNs and huge trajectory (sample) size,
the simulation time cost can become prohibitive.
Intuitively,  the simulation time can be reduced if the $n$-time operations can be speeded up,
for which we propose two solutions.
One is to perform \emph{network reduction} such that the total number of nodes is reduced.
The other is to preform \emph{node-grouping} in order to parallelise the process for checking perturbations, making selections, and updating nodes.
For the first solution, we analyse the PBN structure to identify those nodes that can be removed
and remove them to reduce the network size;
while for the second solution, we analyse the PBN structure to divide nodes into groups and perform
the operations for nodes in a group simultaneously.
We combine the two solutions together
and refer to this simulation technique as \textit{structure-based parallelisation}.
We formalise the two solutions in the following three steps:
the first solution is described in step 1 and the second solution is described in steps 2 and 3.

\smallskip
\begin{itemize}[leftmargin=1.0in]
\item[Step 1.] Remove unnecessary nodes from the PBN.
\item[Step 2.] Parallelise the perturbation process.
\item[Step 3.] Parallelise updating a PBN state with predictor functions.
\end{itemize}
\smallskip
We describe these three steps in the following subsections.

%----------------------------------------------------------------------------------
\subsection{Removing unnecessary nodes}
\label{ssec:leaves}
%----------------------------------------------------------------------------------
We first identify those nodes that can be removed and preform network reudction.
When simulating a PBN without perturbations,
if a node does not affect any other node in the PBN,
the states of all other nodes will not be affected after removing this node.
If this node is not of interest of the analysis, e.g., we are not interested in analysing its steady-state,
then this node is dispensable in a PBN without perturbations.
We refer to such a dispensable node as a leaf node in a PBN and define it as follow:
\begin{definition}[Leaf node]
A~node in a~PBN is a~\emph{leaf node} (or \emph{leaf} for short) if and only if either (1) it is not of interest and has
no child nodes or (2) it is not of interest and has no other children after iteratively removing all its child nodes
which are leaf nodes.
\end{definition}
According to the above definition,
leaf nodes can be simply removed without affecting the simulation of the remaining nodes
in a PBN without perturbations.
In the case of a PBN with perturbations,
perturbations in the leaf nodes need to be considered.
Updating states with Boolean functions will only be performed when
there is no perturbation in both the leaf nodes and the non-leaf nodes.
Perturbations of the leaf nodes can be checked in constant time irrespective of the number of leaf nodes
as describe in Algorithm~\ref{alg:checkleave}.
The input $t$ in this algorithm is the probability that no perturbation happens in all the leaf nodes.
It can be computed easily as $t=(1-p)^\ell$, where $p$ is the perturbation rate for each node and $\ell$ is the number of leaf nodes in the PBN.
With the consideration of their perturbations,
the leaf nodes can be removed without affecting the simulation of the non-leaf nodes
in a PBN with perturbations as well.
Since the leaves are not of interest, results of analyses performed on the simulated trajectories
of the reduced network, i.e.,. containing only non-leaf nodes, will be the same as performed on
trajectories of the original network, i.e., containing all the nodes.

\begin{algorithm}[!t]
\caption{Checking perturbations  of leaf nodes in a PBN}
\label{alg:checkleave}
\begin{algorithmic}[1]
\Procedure{CheckLeafNodes}{$t$}
	\If {$rand()>t$} return ${\it true}$;
	\Else~~return ${\it false}$;
	\EndIf
\EndProcedure
\end{algorithmic}
\end{algorithm}

%----------------------------------------------------------------------------------
\subsection{Performing perturbations in parallel}
\label{ssec:pip}
%----------------------------------------------------------------------------------
The second step of our method speeds up the process for determining perturbations.
Normally, perturbations are checked for nodes one by one.
In order to speed up the simulation of a PBN, we perform perturbations for $k$ nodes simultaneously instead of one by one.
For those $k$ nodes,
there are $2^k$ different perturbation situations.
We compute the probability for each situation and construct an alias table based on the distribution.
With the alias table, we make a choice $c$ among $2^k$ choices
and perturb the corresponding nodes based on the choice.
The choice $c$ is an integer in $[0, 2^k)$
and for the whole network the perturbation can then be performed $k$ nodes by $k$ nodes using the logical bitwise \textit{exclusive or} operation which outputs true only when inputs differ.
To save memory, the alias table can be reused for all the groups since the perturbation rate for each node is the same.
It might happen that the number of nodes in the last perturbation round will be less than $k$ nodes.
Assume there is $k'$ nodes in the last round and $k'< k$.
For those $k'$ nodes, we can reuse the same alias table to make the selection in order to save memory.
After getting the choice $c$, we perform $c=c\&m$, where $\&$ is a~bitwise \textit{and} operation
and $m$ is a~mask constructed by setting the first $k'$ bits of $m$'s binary representation to $1$
and the remaining bits to $0$.
\begin{theorem}
\label{theo:per}
The above process for determining perturbations for the last $k'$ nodes
guarantees that the probability for each of the $k'$ nodes to be perturbed is still $p$.
\end{theorem}
\begin{proof}
Without loss of generality, we assume that in the last $k'$ nodes, $t$ nodes should be perturbed
and the positions of the $t$ nodes are fixed.
The probability for those $t$ fixed nodes to be perturbed is $p^t(1-p)^{k'-t}$.
%The theorem is proved if the process guarantees that the probability of the $t$ fixed nodes to be perturbed is $p^t(1-p)^{k'-t}$.
When we make a selection from the alias table for $k$ nodes,
there are $2^{k-k'}$ different choices corresponding to the case that $t$ fixed position nodes in the last $k'$ nodes are perturbed.
The sum of the probabilities of the $2^{k-k'}$ different choices is
$[p^t(1-p)^{k'-t}]\cdot\sum_{i=0}^{k-k'}p^{i}(1-p)^{k-k'-i} = p^t(1-p)^{k'-t}$.\qed
\end{proof}

We describe the process for constructing groups and performing perturbations based on the groups in Algorithm~\ref{alg:per}.
%The combination of step 1 and step 2 is not demonstrated;
%instead, we will combine the three steps together and demonstrate it in the next subsection.
Algorithm~\ref{alg:per} requires three inputs: $n$ is the number of nodes,\footnote{
In our methods, it is clear that steps 2 and 3 are independent of step 1.
Thus, we consistently use $n$ to denote the number of nodes in a PBN.}
$k$ is the maximum number of nodes that can be perturbed simultaneously
and $s$ is the PBN's current state which is represented by an~integer.
As perturbing one node equals to flipping one bit of $s$,
perturbing nodes in a group is performed via a~logical bitwise \textit{ exclusive or} operation
(see line~\ref{alg-line:perturbe} of Algorithm~\ref{alg:per}).
Perturbing $k$ nodes simultaneously requires $2^k$ double numbers to store the probabilities of
$2^k$ different choices. The size of $k$ is therefore restricted by the available
memory.\footnote{ For the experiments, we set $k$ to 16 and $k$ could be bigger as long as the
memory allows. However, a~larger $k$ requires larger table to store the $2^k$ probabilities
and the performance of a~CPU drops when accessing an~element of a~much larger table
due to the large cache miss rate.}

\begin{algorithm}[!t]
\caption{The group perturbation algorithm}
\label{alg:per}
\begin{algorithmic}[1]
\Procedure{PreparePerturbation}{$n, k$}
	\State $g=\lceil n/k \rceil$; ~$k=\lceil n/g \rceil$;~~$k'=n-k*(g-1)$;
    \State construct the alias table $A_p$ and $mask$;
    \State return $[A_p, mask]$.
\EndProcedure
\Procedure{Perturbation}{$A_p, mask, s$}
    \State $i=0$;~~${\it perturbed=false}$;
    \Repeat
    \State $c=Next(A_p)$; \hfill{\it //$Next(A_p)$ returns a random integer based on $A_p$}
    \If {$c!=0$}
    \State $s=s\oplus(c<\!\!\!<(i*k))$; \hfill{\it //shift c to flip only the bits (nodes) of current group}
    \label{alg-line:perturbe}
    \State ${\it perturbed=true}$;
    \EndIf
    \State $i++$;
	\Until $i=g-1$
	\State $c=Next(A_p)\ \&\ mask$;
	\If {$c!=0$}
    \State $s=s\oplus(c<\!\!\!<(i*k))$; ~~${\it perturbed=true}$;
    \EndIf
    \State return $[s, perturbed]$.
\EndProcedure
\end{algorithmic}
\end{algorithm}

%----------------------------------------------------------------------------------
\subsection{Updating nodes in parallel}
\label{ssec:cpf}
%----------------------------------------------------------------------------------
The last step to speed up PBN simulation is to update a number of nodes simultaneously in
accordance with their predictor functions.
For this step, we need an initialisation process to divide the $n$ nodes
into $m$ groups and compute the combined predictor functions for each group.
After this initialisation, we can select a combined predictor function for each group
based on a sampled random number and apply this combined function to update the nodes in
the group simultaneously.

We first describe how predictor functions of two nodes are combined.
The combination of functions for more than two nodes can be performed iteratively.
Let $x_{\alpha}$ and $x_{\beta}$ be the two nodes to be considered.
Their predictor functions are denoted as
$F_{\alpha}=\{f_1^{({\alpha})}, f_2^{({\alpha})}, \dots, f_{\ell({\alpha})}^{({\alpha})}\}$
and $F_{\beta}=\{f_1^{({\beta})}, f_2^{({\beta})}, \dots, f_{\ell(\beta)}^{({\beta})}\}$.
The corresponding selection probabilities are denoted as
$C_{\alpha}=\{c_1^{({\alpha})},c_2^{({\alpha})},\dots,c_{\ell(\alpha)}^{{\alpha}}\}$
and $C_{\beta}=\{c_1^{({\beta})},c_2^{({\beta})},\dots,c_{\ell({\beta})}^{{\beta}}\}$.
After the grouping,
the number of combined predictor functions is  $\ell(\alpha) * \ell(\beta).$
We denote the set of combined predictor functions as
$\bar{F}_{\alpha\beta}=
\{f_1^{({\alpha})}\cdot f_1^{({\beta})}, f_1^{({\alpha})}\cdot f_2^{({\beta})}, \dots, f_{\ell({\alpha})}^{({\alpha})}\cdot f_{\ell({\beta})}^{({\beta})}\}$,
where for $i \in [1,\ell(\alpha)]$ and $j \in [1,\ell(\beta)]$,
$f_i^{({\alpha})}\cdot f_j^{({\beta})}$ is a combined predictor function
that takes the input nodes of functions $f_i^{({\alpha})}$ and $f_j^{({\beta})}$ as its input and
combines the Boolean output of functions $f_i^{({\alpha})}$ and $f_j^{({\beta})}$ into integers as output.
% i.e., $f_i^{({\alpha})}\cdot f_j^{({\beta})}=2*f_j^{({\beta})}+f_i^{({\alpha})}$.
The combined integers range in $[0,3]$ and their 2-bit binary representations (from right to left) represent the values of nodes $x_{\alpha}$ and $x_{\beta}$.
The selection probability for function $f_i^{({\alpha})}\cdot f_j^{({\beta})}$ is $c_i^{(\alpha)}*c_j^{(\beta)}$.
It holds that $\sum_{i=1}^{\ell(\alpha)}\Sigma_{j=1}^{\ell(\beta)}(c_i^{(\alpha}*c_j^{(\beta)})=1$.
With the selection probabilities, we can compute the alias table for each group so that the
selection of combined predictor function in each group can be performed in constant time.

We now describe how to divide the nodes into groups.
Our aim is to have as few groups as possible so that the updating
of all the nodes can be finished in as few rounds as possible.
However, fewer groups lead to many more nodes in a group,
which will result in a huge number of combined predictor functions in the group.
Therefore, the number of groups has to been chosen properly so that
the number of groups is as small as possible,
while the combined predictor functions can be stored within the memory limit
of the computer performing the simulation.
Besides, nodes with only one predictor function should be considered separately
since selections of predictor functions for those nodes are not needed.
In the rest of this section,
we first formulate the problem for dividing nodes with more than one predictor function
and give our solution afterwards;
then we discuss how to treat nodes with only one predictor function.
%For nodes with more than one predictor function,
%we formulate the problem below and give our solution afterwards.

\medskip\noindent\textbf{Problem description.}
%Assign a unique ID $id_i$, where $i \in [1,n]$, to each node in a PBN with $n$ nodes.
%Let $S$ be a list of $n$ items $\{(id_1,a_1), (id_2,a_2),\dots, (id_n,a_n)\}$.
%For item $(id_i,a_i)$ where $i \in [1,n]$, $id_i$ is the id of this item representing node $id_i$ in a PBN
%and $a_i$ is the weight of this item representing the number of predictor functions of node $id_i$.
%Assign a distinct ID to any $a_i \in S$ and denote it as $\bar{a}_i$.
%$\bar{S}=\{\bar{a}_1,\bar{a}_2,\dots,\bar{a}_n\}$ becomes a set.
Let $S$ be a list of $n$ items $\{\mu_1,\mu_2,\dots,\mu_n\}$.
For $i \in [1,n]$,  item $\mu_i$ represents a node in a PBN with $n$ nodes
and its weight is assigned by a function $\omega(\mu_i)$,
which returns the number of predictor functions of node $\mu_i$.
We aim to find a minimum integer $m$ to distribute the nodes into $m$ groups
such that the sum of the combined predictor functions numbers of the $m$ groups will not exceed
a~memory limit $\theta$.
This is equivalent to finding a~minimum $m$ and an~$m$-partition $S_1, S_2, \dots S_m$ of $S$,
i.e., $S=S_1 \cup S_2 \cup \dots \cup S_m$ and $S_k \cap S_\ell =\emptyset$ for $k,\ell \in \{1,2,\dots,m\}$,
such that $\sum_{i=1}^{m}\left(\Pi_{\mu_j \in S_i} \omega(\mu_j)\right) \leq \theta$.

\medskip\noindent\textbf{Solution.}
The problem in fact has two outputs:
an integer $m$ and an $m$-partition.
We first try to estimate a potential value of $m$,
i.e., the lower bound of $m$ that could lead to an $m$-partition of $S$
%$S=S_1 \cup S_2 \cup \dots \cup S_m$ of $S$
%such that $S_k \cap S_l =\emptyset$ for $k,l \in \{1,2,\dots,m\}$
which satisfies $\sum_{i=1}^{m}\left(\Pi_{\mu_j \in S_i} \omega(\mu_j)\right)  \leq \theta$.
With this estimate, we then try to find an $m$-partition satisfying the above requirements.

Denote the weight of a sub-list $S_i$ as $w_i$
and $w_i=\Pi_{\mu_j \in S_i} \omega(\mu_j).$
The inequality in the problem description can be rewritten as $\sum_{i=1}^{m}w_i \leq \theta$.
We first compute the minimum value of $\hat{m}$ satisfying the following inequality and denote this minimum value as $\hat{m}_{min}$.
\begin{equation}
\label{average}
\hat{m}\cdot \sqrt[\hat{m}]{\Pi_{i=1}^{n}\omega(\mu_i)} \leq \theta.
\end{equation}
\begin{theorem}
\label{theo:min_inequlity}
$\hat{m}_{min}$ is the lower bound of $m$ that allows a partition to satisfy $\sum_{i=1}^{m}w_i \leq \theta$.
\end{theorem}
\begin{proof}
This is equivalent to proving that
\begin{equation}
\label{min_inequlity}
\sum_{i=1}^{\hat{m}_{min}-1}w_i^{'} > \theta,
\end{equation}
 where $w_{i}^{'}$ is the weight of the $i$th sub-list
in an arbitrary partition of $S$ into $\hat{m}_{min}-1$ sub-lists.
Since $\hat{m}_{min}$ is the minimum value of $\hat{m}$ that satisfies Inequality~\eqref{average},
we have $(\hat{m}_{min}-1)\cdot \sqrt[(\hat{m}_{min}-1)]{\Pi_{i=1}^{n}\omega(\mu_i)} > \theta$.
Hence,
\begin{equation}
\label{min_inequlity2}
(\hat{m}_{min}-1)\cdot \sqrt[(\hat{m}_{min}-1)]{\Pi_{i=1}^{\hat{m}_{min}-1}w_i^{'}} > \theta.
\end{equation}
Based on the inequality of arithmetic and geometric means, % (see Theorem~\ref{theo:ineq_arith}),
we have
\begin{equation}
\label{arithmetic}
\sum_{i=1}^{\hat{m}_{min}-1}w_i^{'} \geq (\hat{m}_{min}-1)\cdot \sqrt[(\hat{m}_{min}-1)]{\Pi_{i=1}^{\hat{m}_{min}-1}w_i^{'}}.
\end{equation}
Inequality~\eqref{min_inequlity} follows from Inequality~\eqref{min_inequlity2} and Inequality~\eqref{arithmetic}.\qed
\end{proof}

%\quad
%The inequality of arithmetic and geometric means is describe in the following theorem.
%\begin{theorem}
%\label{theo:ineq_arith}
%For any list of $n$ nonnegative real numbers $x_1,x_2,\dots,x_n$,
%it holds that
%\begin{equation}
%\label{equ:ineq_arith}
%\frac{1}{n}\sum_{i=1}^{n}x_i \geq \sqrt[n]{\Pi_{i=1}^{n}x_i},
%\end{equation}
%and that equality holds if and only if $x_1=x_2=\dots=x_n.$
%\end{theorem}

\begin{algorithm}[!t]
\caption{The greedy algorithm}
\label{alg:greedy}
\begin{algorithmic}[1]
\Procedure{FindPartitions}{$S, m$}
	\State sort $S$ with descending orders based on the weight of items in $S$;
    \State initialise $A$, an array of $m$ lists;~ $j=0$;\hfill{\it //initially, each $A[i]$ is an empty list }
    \Repeat \hfill{\it //$S[j]$ ($j$ starts from 0) is the $j$th item in $S$ }
%    \State $i=0$; ~$index=0$;
%    \Repeat
%    \If {$min \geq {\it multip}(A[i])$}  $min={\it multip}(A[i])$; ~$index=i$;
%    \EndIf  \hfill{\it //see Equation~\ref{equ:multip} for the definition of $multip$}
%    \State $i++$;
%    \Until $i=m-1$;
%    \State add $S[j]$ into $A[i]$;  \hfill{\it //$S[j]$ ($j$ starts from 0) is the $j$th item in $S$ }
%    \State $min={\it MAX\_INTEGER}$;~$j++$;
%     \label{lst:M_N}
	\State among the $m$ elements of $A$,
	\State find the one with the smallest total weight and add $S[j]$ to it;
	\State$j++$;
	\Until{$j=S.size()$} \hfill{\it //$S.size()$ returns the number of items in $S$ }
	\State return $A$.
\EndProcedure
\end{algorithmic}
\end{algorithm}

Starting from the lower bound, we try to find a partition of $S$ into $m$ sub-lists that satisfies $\sum_{i=1}^{m}w_i \leq \theta$.
%Since the equality in Equation~\ref{equ:ineq_arith} holds if and only if $x_1=x_2=\dots=x_n$,
%we get the heuristic that the weight of the $m$ subsets should be as equal as possible.
Since the arithmetic and geometric means of non-negative real numbers
are equal if and only if every number is the same,
we get the heuristic that the weight of the $m$ sub-lists should be as equal as possible
so that the sum of the weights is as small as possible.
Our problem then becomes similar to the NP-hard multi-way number partition problem:
to divide a given set of integers into a collection of subsets,
so that the sum of the numbers in each subset are as nearly equal as possible.
%The difference is that our problem is based on multiset and requires that the multiplications of the numbers in partitions are as nearly equal as possible.
We adapt the greedy algorithm (see Algorithm~\ref{alg:greedy} for details)
for solving the multi-way number partition problem,
by modifying the sum to multiplication, in order to solve our partition problem.\footnote{
There exist other algorithms to solve the multi-way number partition problem,
and we choose the greedy algorithm for its efficiency.}
%The ${\it multip}$ function used in this algorithm is defined as follow:
%\begin{equation}
%{\it multip}(A[i])=\begin{cases}
%    \Pi_{j=0}^{A[i].size()-1}\omega(A[i][j])      & \quad \text{if } A[i]\text{~is not empty}\\
%    0 & \quad \text{otherwise}\\
%  \end{cases}
%  \label{equ:multip}
%\end{equation}
%where $S.size()$ is the number of elements in the list $S$.
%
%Alternatively, the Karmarkar–Karp heuristic~\cite{KK82} can be used to approximate the partition as well.
%The Karmarkar-Karp heuristic has the same complexity but often results in better results than the greedy algorithm.
%We describe the idea of the Karmarkar–Karp heuristic for the 2-way partition.
%The heuristic runs with 2 phases.
%The first phase of the heuristic takes the two largest numbers from the input and replaces them by their difference;
%this is repeated until only one number remains.
%The replacement represents the decision to put the two numbers in different sets,
%without immediately deciding which one is in which set.
%At the end of phase one, the single remaining number is the difference of the two subset sums.
%The second phase reconstructs the actual solution.
%We modify the above two phases by replacing the difference with the ratio,
%changing set to multiset,
%and at the end of phase one, the single remaining number is the ratio of the two sub-multiset multiplications.
%
If the $m$-partition we find satisfies the requirement $\sum_{i=1}^{m}w_i \leq \theta$,
then we get a solution to our problem.
%the potential $m$ we find satisfies the requirements and we get the solutions.
Otherwise, we need to increase $m$ by one and try to find a new $m$-partition.
We repeat this process until the condition $\sum_{i=1}^{m}w_i \leq \theta$ is satisfied.
The whole partition process for all the nodes is described in Algorithm~\ref{alg:partition}.

Nodes with only one predictor function are treated in line~\ref{alg-line:one}.
We divide such nodes into groups based on their parent nodes,
i.e., we put nodes sharing the most common parents into the same group.
In this way, the combined predictor function size can be as small as possible such that
the limited memory can handle more nodes in a group.
The number of nodes in a group is also restricted by the combined predictor function size,
i.e., the number of parent nodes in this group.\footnote{
In our experiments, the maximum number of parent nodes in one group is set to 18.
Similar to the value of $k$ in step 2, the number can be larger as long as the memory can handle.
However, the penalty from large cache miss rate will diminish the benefits by having fewer groups
when the number of parent nodes is too large.}
The partition is performed with an algorithm similar to Algorithm~\ref{alg:greedy}.
The difference is that in each iteration we always add a node into a group which shares most common parent nodes with this node.

\begin{algorithm}[!t]
\caption{Partition $n$ nodes into groups.}
\label{alg:partition}
\begin{algorithmic}[1]
\Procedure{Partition}{$G, \theta$}%\hfill{\it //$k'$ is the maximum number of parent nodes in a group}
	\State compute two lists $S$ and $S'$ based on $G$; \hfill{\it //$S'$ contains nodes with one predictor function}
	\State compute the lower bound $\hat{m}$; ~ $m=\hat{m}$;%and $\hat{m}'$; \hfill{\it //$\hat{m}'$ is computed with $k'$}
	\label{alg-line:seperate}
	%\State $m=\hat{m}-1$;% $m'=\hat{m}'-1$;
	\Repeat
	\State $A_1=$ {\sc FindPartitions}$(S,m)$;
	\State compute the sum of the weight of the $m$ partitions and assign it to $sum$;
	\State $m=m+1$;
	\Until $sum <\theta$
%	\Repeat
%	\State $m'=m'+1$;~$[done,A_2]=${\sc PartitionOneFunction}$(S',m')$;
%	\Until $done=true$;
	\State divide $S'$ into $A_2$; \hfill{\it //partition nodes with only one predictor function}
	\label{alg-line:one}	
	\State merge $A_1$ and $A_2$ into $A$;
	\State return $A$.
\EndProcedure
\end{algorithmic}
\end{algorithm}

%----------------------------------------------------------------------------------
\subsection{The new simulation method}
\label{ssec:method}
%----------------------------------------------------------------------------------
We describe our new method for simulating PBNs in Algorithm~\ref{alg:slp}.
The procedure {\sc Preparation} describes the whole preparation process of the three steps
(network reduction for Step~1, and node-grouping for Step~2 and Step~3). The four inputs of the
procedure {\sc Preparation} are respectively the PBN network $G$,
the memory limit $\theta$, the maximum number $k$ of nodes that can be put in a group for perturbation and the maximum number of parent nodes in a group.
The {\sc Preparation}  procedure takes these inputs,
performs network reduction and node grouping.
The reduced network and the grouped nodes information are then provided
for the {\sc ParallelSimulation} procedure via six parameters:	
$A_p$ and $mask$ are the alias table and mask used for performing perturbations of non-leaf nodes
as explained in Algorithm~\ref{alg:per}; $t$ is used for checking perturbations in leaf nodes as
explained in Algorithm~\ref{alg:checkleave};
$A$ is an array containing the alias tables for predictor functions in all groups;
$F$ is an array containing predictor functions of all groups;
and $cum$ is an array storing the cumulative number of nodes in each group.
Perturbations for leaf nodes and non-leaf nodes have been explained in
Algorithms~\ref{alg:checkleave} and \ref{alg:per}.
We now explain how nodes in a~group are simultaneously updated with combined predictor function.
It is performed via the following three steps:
 1) a random combined predictor function is selected from $F$ based on the alias table $A$;
 2) the output of the combined predictor function is obtained according to the current state $s$;
 3) the nodes in this group are updated based on the output of the combined predictor function.
To save memory, states are stored as integers and updating a~group of nodes is implemented via
a~logical bitwise \textit{or} operation. To guarantee that the update is performed on the required
nodes, a~shift operation is needed on the output of the selected function (line~\ref{alg-line:cum}).
The number of bits to be shifted for the current group is in fact the cumulative number of nodes of all its previous groups,
which is stored in the array $cum$.

%Note that the index of $cum$ starts with 0 and $cum[0]$ is set to 0.
%
%$cum$ is needed so that the update of nodes in each group can be performed correctly on the desired nodes
%(see line~\ref{alg-line:cum} of Algorithm~\ref{alg:slp}).
%The {\sc ParallelSimulation} procedure takes the outputs of the {\sc Preparation} procedure,
%together with the current state $s$ as its input,
%performs simulation for nodes in groups simultaneously for the reduced network.

\begin{algorithm}[!t]
\caption{Structure-based PBN simulation.}
\label{alg:slp}
\begin{algorithmic}[1]
\Procedure{Preparation}{$G,\theta, k$}
		\State perform network reduction for $G$ and store the reduced network in $G'$
		\State get the number of nodes $n$ and perturbation rate $p$ from $G$;
		\State get the number of nodes $n'$ from $G'$; ~$t=pow(1-p,n-n');$
		\State $[A_p, mask]=${\sc PreparePerturbation}$(n',k)$;
		\State $PA=${\sc Partition}$(G',\theta)$;
		\State for each group in $PA$, compute its combined functions $F$, and its alias table $A$,
		\State and the cumulative number of nodes in the group $cum$;
		\State return $[A_p, mask, t, A, F, cum]$.
\EndProcedure
\Procedure{ParallelSimulation}{$A_p, mask, A, F, cum, t, s$}
	\State $[s, {\it perturbed}]=${\sc Perturbation}$(A_p,mask,s)$; \hfill{\it //perform perturbations by group}
	\label{alg-line:perturbaion}
    \If{${\it perturbed}~ ||$ {\sc CheckLeafNodes}$(t)$} ~~return $s$; \hfill{\it //check perturbations of leaves}
    \label{alg-line:checkleaf}
    \Else~~$i=0$; $s'=0;$ ${\it count=size(A)};$ \hfill{\it //$size(A)$ returns the number of elements in array $A$}
    \Repeat
    \State $index=Next(A[i]);$ \hfill{\it //select a random integer based on the alias table of group i}
    \State$f=F.get({\it index})$; \hfill{\it //obtain the predictor function at the given index}
    \label{alg-line:start}
    \State $v=f[s]$; \hfill{\it //$f[s]$ returns the integer output of $f$ based on state $s$}
    \State $s'=s'~|~(v<\!\!\!<cum[i]);$ \hfill{\it //shift $v$ to update only nodes in the current group}
    \label{alg-line:cum}
    \State $i++$;
    \Until $i={\it count-1}$;
    \EndIf
	\State return $s'$.
\EndProcedure
\end{algorithmic}
\end{algorithm}

%===============================================
\section{Evaluation}
\label{sec:evaluation}
%===============================================
The evaluation of our new simulation method is performed on both randomly generated networks and a real-life biological network.
All the experiments are performed on a high performance computing (HPC) machines,
each of which contains a CPU of Intel Xeon X5675 @ 3.07 GHz.
The program is written in Java and the initial and maximum Java virtual machine heap size is set to 4GB and 5.89GB, respectively.
The evaluation data is available at \url{http://satoss.uni.lu/software/ASSA-PBN/benchmark}.
%----------------------------------------------------------------------------------
\subsection{Randomly generated networks}
\label{ssec:ernd}
%----------------------------------------------------------------------------------
%The evaluation on randomly generated networks consists of two experiments.
%The first experiment focuses on evaluating the effectiveness of the methods on thousands of randomly generated networks.
%The second experiment focuses on evaluating the influence of sample size on the performance of our method.
With the evaluation on randomly generated networks, we aim not only to show the efficiency of our
methods, but also to answer how much speedup our method is likely to provide for a~given PBN.

The first step of our new simulation method performs a~network reduction technique,
which is different from the node-grouping technique in the later two steps.
Therefore, we evaluate the contribution of the first step and the other two steps to the performance
of our new simulation method separately.
We name the original simulation method as {\sf $Method_{old}$};
the simulation method applying the network reduction technique as {\sf $Method_{reduction}$};
and the simulation method applying both the network reduction and node-grouping techniques as {\sf $Method_{new}$}.
{\sf $Method_{reduction}$} and {\sf $Method_{new}$} require pre-processing of the PBN under study,
which leads to a~certain computational overhead. However, the proportion of the pre-processing
time in the whole computation decreases with the increase of the sample size. In our evaluation,
we first focus on comparisons without taking pre-processing into account to evaluate the maximum
potential performance of our new simulation method;
%generally how efficient is our new simulation method;
we then show how different sample sizes will affect the performance when pre-processing is
considered.

\medskip\noindent\textbf{How does our method perform?}
Intuitively, the speedup due to the network reduction technique is influenced by how much
a~network can be reduced and the performance of node-grouping is influenced by both the density
and size of a~given network.
Hence, the evaluation is performed on a~large number of randomly
generated PBNs covering different types of networks.
%with respect to their network size, density, and percentage of leaf nodes.
In total, we use 2307 randomly generated PBNs with different percentages of leaves ranging between
0\% and 90\%; different densities ranging between 1 and 8.1; and different network sizes from the
set $\{20, 50, 100, 150, 200, 250, 300, 350, \\400, 450, 500, 550, 600, 650, 700, 750, 800, 850,
900, 950, 1000\}$. We simulate 400 million steps for each of the 2307 PBNs with the three
different simulation methods and compare their time costs. For the network reduction technique
the speedups are computed as the ratio between the time of {\sf $Method_{reduction}$} and the time
of {\sf $Method_{old}$}. The obtained speedups are between 1.00 and 10.90. For node-grouping,
the speedups are computed as the ratio between the time of {\sf $Method_{new}$} and the time of
{\sf $Method_{reduction}$}. We have obtained speedups between 1.56 and 4.99. We plot in
Figure~\ref{fig:speedups} the speedups of the network reduction and node-grouping techniques
with respect to their related parameters.
For the speedups achieved with network reduction, the related parameters are  the percentage of leaves and the density.
In fact, there is little influence from density to the speedup resulting from network reduction
as the speedups do not change much with the different densities (see Figure~\ref{fig:speedup1}).
The determinant factor is the percentage of leaves.
The more leaves a PBN has, the more speedup we can obtain for the network.
For the speedups obtained from node-grouping, the related parameters are the density and the network size after network reduction,
i.e., the number of non-leave nodes.
Based on Figure~\ref{fig:speedup2}, the speedup with node-grouping is mainly determined by the
network density: a~smaller network density could result in a~larger speedup contributed from the
node-grouping technique.
This is mainly due to the fact that sparse network has a relatively small number of predictor functions in each node and therefore, the nodes will be partitioned into fewer groups.
Moreover, while the performance of network reduction is largely
influenced by the percentage of leaves, the node-grouping technique tends to provide a~rather
stable speedup. Even for large dense networks, the technique can reduce the time cost almost by
half.

The combination of these two techniques results in speedups (time of {\sf $Method_{new}$} over
time of {\sf $Method_{old}$}) between 1.74 and 41.92. We plot in Figure~\ref{fig:speedup3} the
speedups in terms of the percentage of leaves and density. The figure shows a~very good
performance of our new method on sparse networks with large percentage of leaves.
\begin{figure}[!t]
  \centering
  \begin{subfigure}[b]{0.48\textwidth}
    \centering
    \includegraphics[scale=0.31]{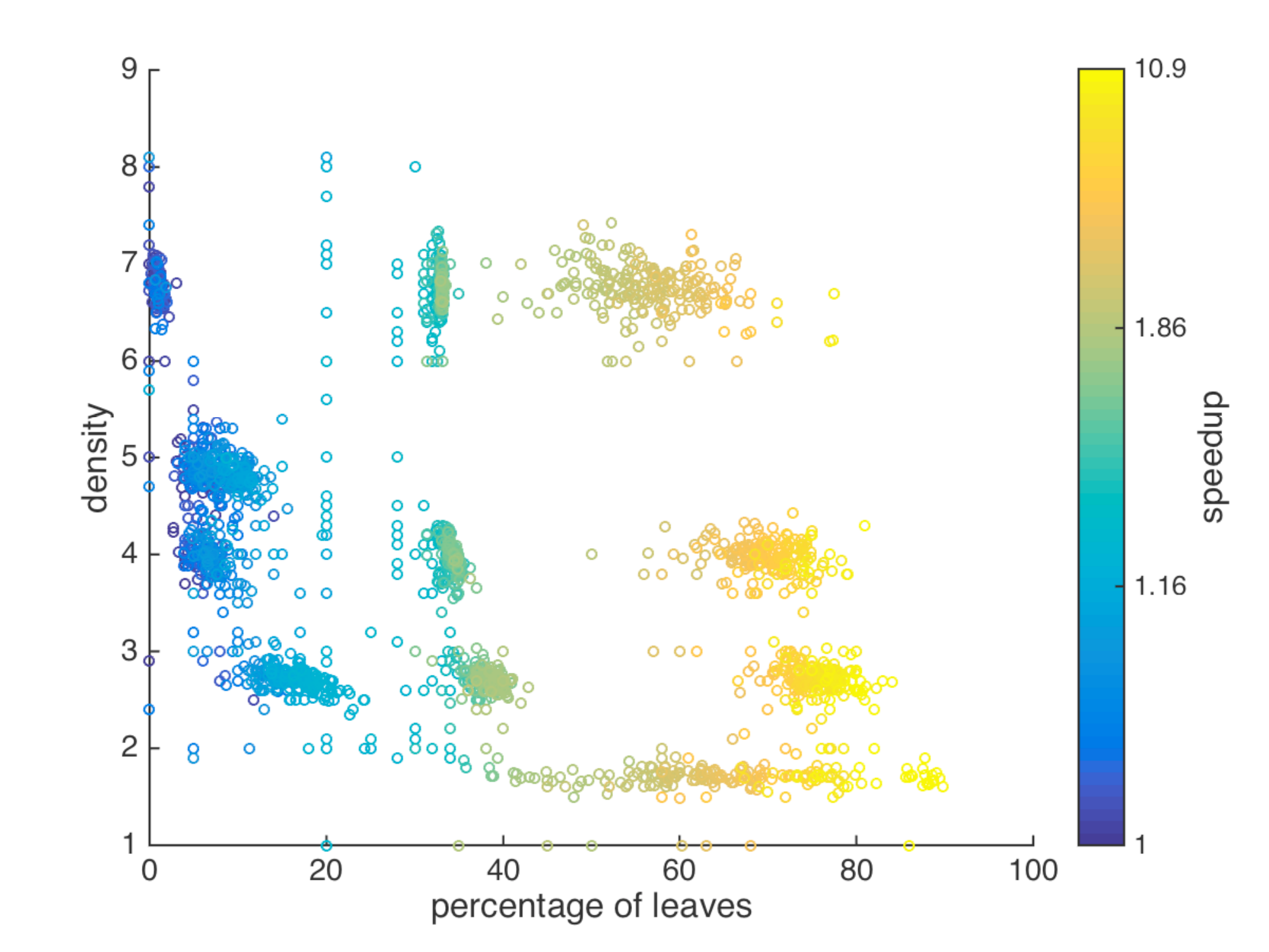}
    \caption{ {\sf $Method_{reduction}$} over {\sf $Method_{old}$}.}
    \label{fig:speedup1}
  \end{subfigure}%
  \quad
  \begin{subfigure}[b]{0.48\textwidth}
    \centering
    \includegraphics[scale=0.31]{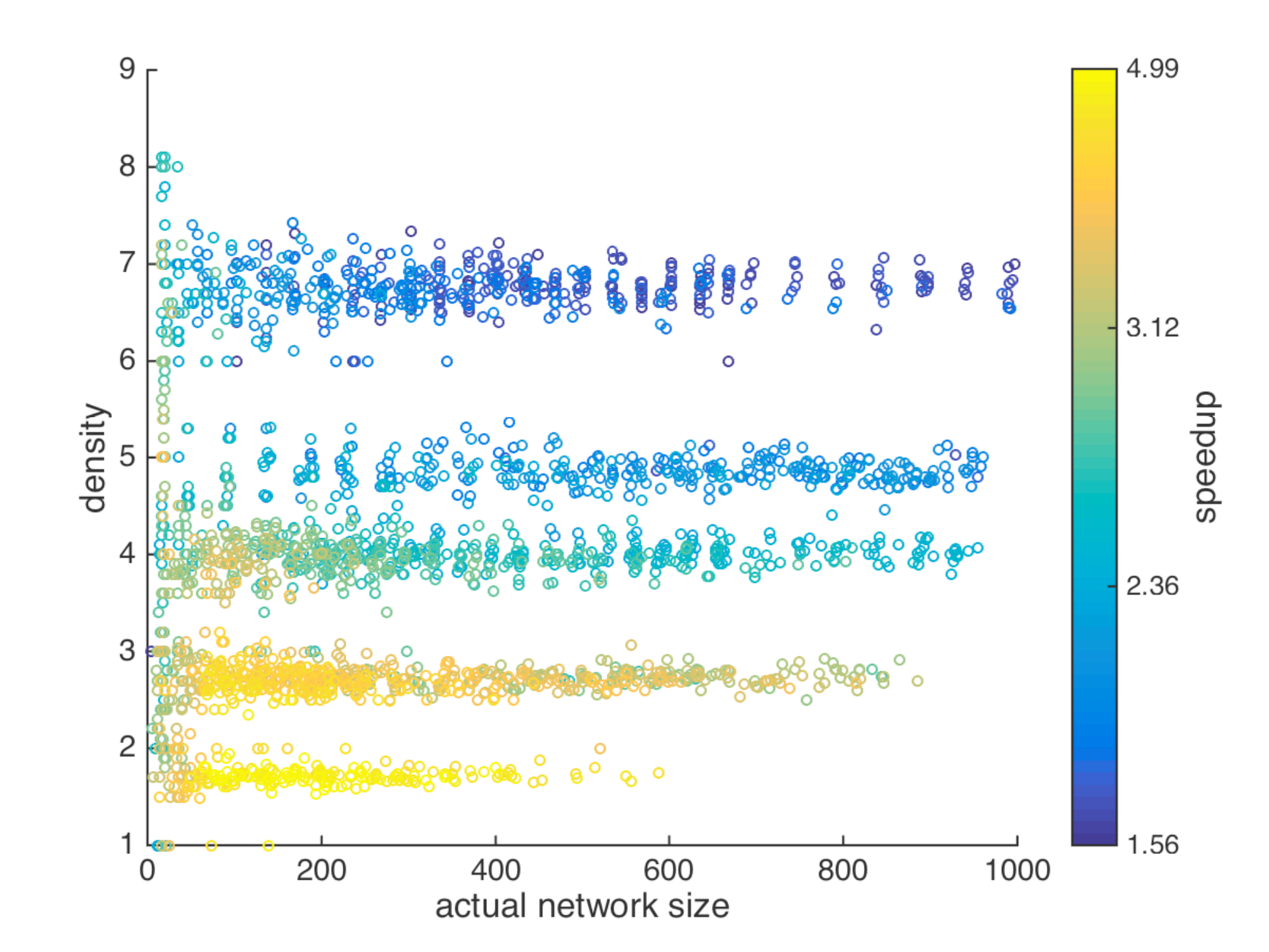}
    \caption{{\sf $Method_{new}$} over {\sf $Method_{reduction}$}.}
    \label{fig:speedup2}
  \end{subfigure}
  \caption{Speedups contributed from network reduction and node-grouping.}
  \label{fig:speedups}
\end{figure}

\begin{figure}[!t]
\centering
\includegraphics[width=0.58\textwidth]{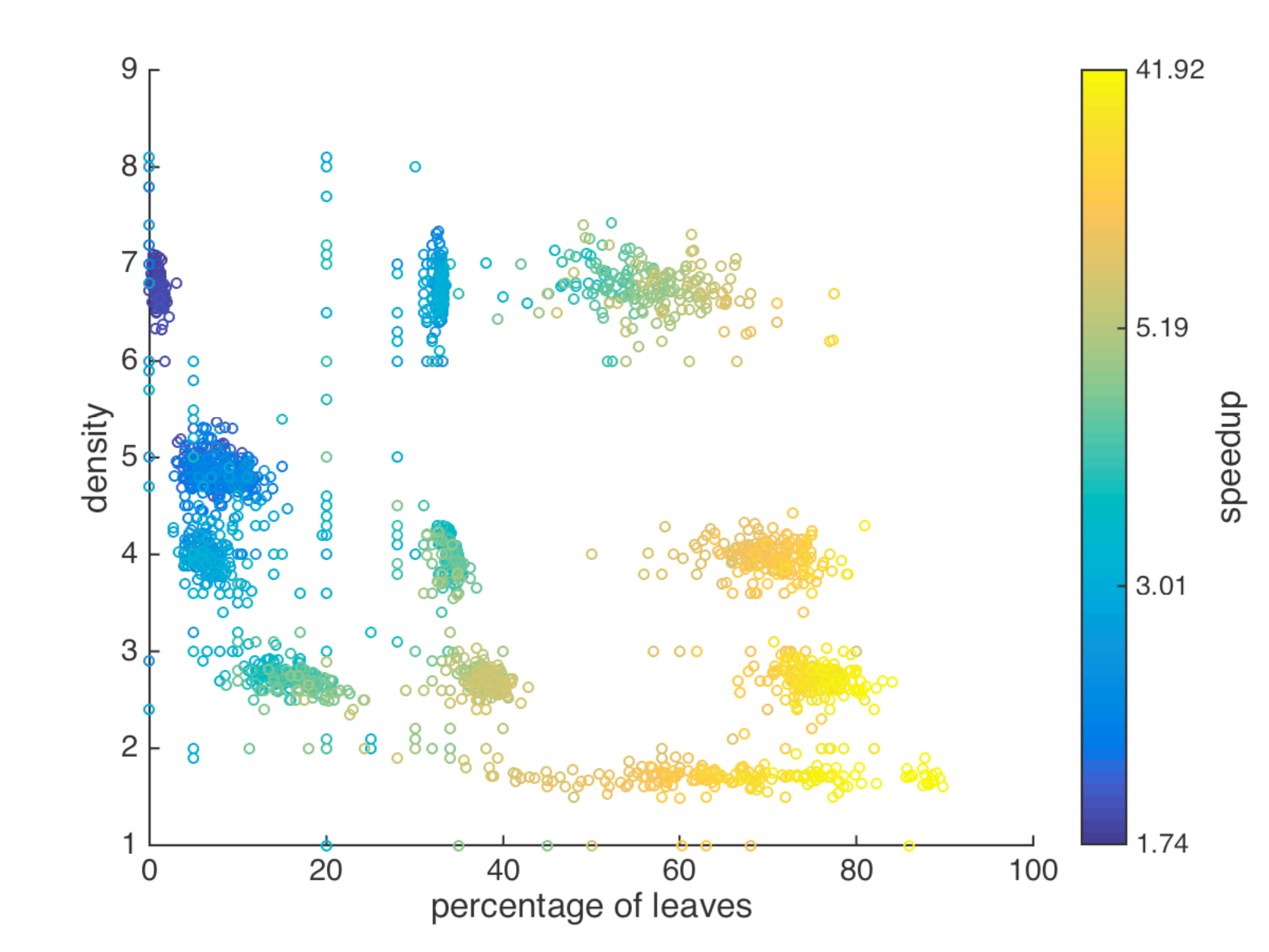}
\caption{Speedups of {\sf $Method_{new}$} over {\sf $Method_{old}$}.}
\label{fig:speedup3}
\end{figure}

\medskip\noindent\textbf{What is the influence of sample size?}
We continue to evaluate the influence of sample size on our proposed new PBN simulation method.
The pre-processing time for the network reduction step is relatively very small;
therefore our evaluation focuses on the influence to the total time cost of all the 3 steps,
i.e., the speedup of {\sf $Method_{new}$} with respect to {\sf $Method_{old}$}.
We selected 9 representative PBNs from the above 2307 PBNs,
with respect to their densities, percentages of leaves and the speedups we have obtained.
%i.e., 3 PBNs from which the minimum three values of speedups are obtained;
%3 PBNs from which the maximum three values of speedups are obtained;
%and 3 PBNs from which the medium three values of speedups are obtained.
We simulate the 9 PBNs for different sample size
using both {\sf $Method_{old}$} and {\sf $Method_{new}$}.
We show the average pre-processing time of {\sf $Method_{new}$}
and the obtained speedups of {\sf $Method_{new}$} (taking into account pre-processing time costs)
with different sample sizes in Table~\ref{tab:influence}.
As expected, with the increase of the sample size, the influence of pre-processing time becomes
smaller and the speedup increases. In fact, in some cases, the pre-processing time is relatively
so small that its influence becomes negligible, e.g., for networks 7 and 8 when the sample size is
equal or greater than 100 million.
%In addition, the density of the network has a strong influence to the preparation time.
%A dense network has more predictor functions in average comparing to a sparse one;
%leading to much more combined predictor functions.
%Therefore, the preparation time for dense network is in general longer than that of sparse network.
Moreover, often with a sample size larger than 10 million,
the effort spent in pre-processing
can be compensated by the saved sampling time (simulation speedup).

\begin{table}[!t]
\centering
\begin{tabular}{|r|r|r|r|r|r|r|r|r|}
\hline
\multicolumn{1}{|c|}{\multirow{3}{*}{network \#}} & \multicolumn{1}{c|}{\multirow{3}{*}{size}} & \multicolumn{1}{c|}{\multirow{3}{*}{\begin{tabular}[c]{@{}c@{}}~percentage~\\ of leaves\end{tabular}}} & \multicolumn{1}{c|}{\multirow{3}{*}{~density~}} & \multicolumn{1}{c|}{\multirow{3}{*}{\begin{tabular}[c]{@{}c@{}}~pre-processing~\\ time (second)\end{tabular}}} & \multicolumn{4}{c|}{\multirow{2}{*}{\begin{tabular}[c]{@{}c@{}}speedup with different\\ sample sizes (million)\end{tabular}}} \\
\multicolumn{1}{|c|}{}                    & \multicolumn{1}{c|}{}                      & \multicolumn{1}{c|}{}                                                                             & \multicolumn{1}{c|}{}                         & \multicolumn{1}{c|}{}                                                                                & \multicolumn{4}{c|}{}                                                                                                      \\ \cline{6-9}
\multicolumn{1}{|c|}{}                    & \multicolumn{1}{c|}{}                      & \multicolumn{1}{c|}{}                                                                             & \multicolumn{1}{c|}{}                         & \multicolumn{1}{c|}{}                                                                                & \multicolumn{1}{c|}{1}      & \multicolumn{1}{c|}{10}      & \multicolumn{1}{c|}{100}      & \multicolumn{1}{c|}{400}      \\ \hline\hline
1                                         & 900                                        & 1.11                                                                                              & 6.72                                          &                                                   28.12                                                   & 0.65                       & 1.49                        & 1.71                         & 1.73                          \\ \hline
2                                         & 950                                        & 0.84                                                                                              & 6.96                                          &                                       32.35                                                               & 0.59                        &        1.47                    &         1.73 & 1.75                          \\ \hline
3                                         & 1000                                       & 0.30                                                                                              & 7.00                                          &     33.72                                                                                                 & 0.58                       & 1.45                         & 1.71                         & 1.73                           \\ \hline
4                                         & 600                                        & 67.83                                                                                             & 4.25                                          &                           162.21                                                                           & 0.13                        & 1.08                         & 4.51                          & 6.89                          \\ \hline
5                                         & 800                                        & 68.38                                                                                             & 3.94                                          &                                                                   43.17                                   & 0.66                        & 3.05                         &         6.75                      & 7.69                          \\ \hline
6                                         & 900                                        & 68.00                                                                                             & 3.89                                          &                                           36.58                                                           & 0.69                        & 3.56                         & 6.90                          & 7.70                          \\ \hline
7                                         & 450                                        & 89.78                                                                                             & 1.60                                          &                       0.23                                                                               & 21.44                       & 37.59                        & 41.62                         & 41.84                         \\ \hline
8                                         & 550                                        & 88.55                                                                                             & 1.72                                          &                                  0.24                                                                    & 20.26                       & 35.94                        & 36.47                         & 36.62                         \\ \hline
~~~~~~9                      & ~~~~1000                                       & 89.10                                                                                             & 1.75                                          &                               1.08                                                                       & ~~~~10.04                       & ~~~~31.83                        &        ~~~~35.09                       & ~~~~37.19                         \\ \hline
\end{tabular}
\caption{Influence of sample sizes on the speedups of {\sf $Method_{new}$} over {\sf $Method_{old}$}.}
\label{tab:influence}
\end{table}

\medskip\noindent\textbf{Performance prediction.}
To predict the speedup of our method for a given network,
we apply regression techniques on the results of the 2307 PBNs to fit a prediction model.
We use the normalised percentage of leaves and the network density
as the predictor variables and the speedup of {\sf $Method_{new}$} over {\sf $Method_{old}$}
as the response variables in the regression model since network size does not directly affect the speedup based on the plotted pictures.
In the end, we obtained a polynomial regression model shown in Equation~\ref{equ:predict}, which can fit $90.9\%$ of the data:
\begin{equation}
y=b_1+b_2*x_1+b_3*x_1^2+b_4*x_2+b_5*x_2^2,
\label{equ:predict}
\end{equation}
where $[b_1,b_2,b_3,b_4,b_5]=[2.89, 2.71, 2.40, -1.65, 0.71]$, $y$ represents the speedup,
$x_1$ represents the percentage of leaves and $x_2$ represent the network density.
The result of a 10-fold cross-validation of this model supports this prediction rate;
hence we believe this model does not overfit the given data.
Based on this model, we can predict how much speedup is likely to be obtained with our proposed method for a given PBN.

\subsection{An apoptosis network}
\label{ssec:erbn}
%----------------------------------------------------------------------------------
In this section,
we evaluate our method on a real-life biological network,
i.e., an apoptosis network containing 96 nodes~\cite{SSVSSBEMS09}.
This network contains $37.5\%$ of leaves and has a density of 1.78,
which is suitable for applying our method to gain speedups.
The apoptosis network has been analysed in~\cite{MPY15}.
In one of the analyses, i.e., the long-term influences~\cite{SDKZ02} on complex2 from
each of its parent nodes: RIP-deubi, complex1, and FADD,
7 steady-state probabilities of the apoptosis network need to be computed.
In this evaluation, we compute the 7 steady-state probabilities using
our proposed structure-based simulation method ({\sf $Method_{new}$})
and compare it with the original simulation method ({\sf $Method_{old}$}).
The precision and confidence level of all the computations,
as required by the two-state Markov chain approach~\cite{RL92}, are set to $10^{-5}$ and 0.95, respectively.
The results of this computation are shown in Table~\ref{tab:case-study}.
The sample sizes required by both methods are very close for computing same steady-state
probabilities.
Note that the speedups are computed based on the accurate data, which are slightly different from
the truncated and rounded data shown in the table.
We have obtained speedups ({\sf $Method_{new}$} over {\sf $Method_{old}$}) between 7.67 and 10.28 for computing those 7 probabilities.
In total, the time cost is reduced from 1.5 hours to about 10 mins.
%===============================================
\section{Discussion and Conclusion}
\label{sec:conclusion}
%===============================================
In this work, we propose a structure-based method for speeding up simulations of PBNs.
Using network reduction and node-grouping techniques,
our method can significantly improve the simulation speed of PBNs.
We show with experiments that our new simulation method is especially efficient in the case of
analysing sparse networks with a~large number of leaf nodes.

The node-grouping technique gains speedups by using more memory. Theoretically, as long as the
memory can handle, the group number can be made as small as possible. However, this causes two
issues in practice. First, the pre-processing time increases dramatically with the group number
decreasing. Second, the performance of the method drops a~lot when operating on large memories due
to the increase of cache miss rate. Therefore, in our experiments we do not explore all the
available memory to maximise the groups. Reducing the pre-processing time cost and the cache miss
rate would be two future works to further improve the performance of our method.
We also plan to apply our method for the analysis of real-life large biological networks.

\begin{table}[!t]
\centering
\begin{tabular}{|c|r|r||r|r|r||r|}
\hline
{\multirow{3}{*}{probability \#}} &\multicolumn{2}{c||}{{\sf $Method_{old}$}}                                                                                                   & \multicolumn{3}{c||}{{\sf $Method_{new}$}}                                                                                                 & \multicolumn{1}{c|}{\multirow{3}{*}{~speedup~}} \\ \cline{2-6}
 & \multicolumn{1}{c|}{\begin{tabular}[c]{@{}c@{}}~sample size~\\ (million)\end{tabular}} & \multicolumn{1}{c||}{time (minute)} & \multicolumn{1}{c|}{\begin{tabular}[c]{@{}c@{}}pre-processing \\time (second) \end{tabular}} & \multicolumn{1}{c|}{\begin{tabular}[c]{@{}c@{}}sample size\\ (million)\end{tabular}} & \multicolumn{1}{c||}{\begin{tabular}[c]{@{}c@{}}total\\ time (minute)\end{tabular}} & \multicolumn{1}{c|}{}                          \\ \hline\hline
1           & 147.50	                                                                               & 9.51	                         & 4.57           & 147.82	                                                                              &1.05	                      & 9.09                                      \\ \hline
2           & 452.35	                                                 & 28.65	     & 3.10           & 452.25	   &2.79	 & 10.28                                   \\ \hline
3           &253.85	                 & 14.88	                     & 3.42          & 253.99	       & 1.74	                & 8.54       \\ \hline
4           & 49.52	    & 2.96	             & 3.38           &50.39	                                                                        & 0.36	 & 8.31       \\ \hline
5           & 315.06	                                                                   &17.73	            & 4.40         &305.43    & 	2.05 & 	8.39       \\ \hline
6           & 62.22	                                      & 3.69	         & 3.13          & 50.28	           & 0.39	            & 7.67     \\ \hline
7           & 255.88	                  &16.74	              & 4.01           & 256.61	      & 1.70	             & 9.88       \\ \hline
\end{tabular}
\caption{Performance of {\sf $Method_{old}$} and {\sf $Method_{new}$} on an apoptosis network.}
\label{tab:case-study}
\end{table}

%===============================================
%----------The Bibliography--------------------
%===============================================
\bibliographystyle{splncs}
\bibliography{structure}

\end{document}